\documentclass[jpa,aps,floatfix,amsmath,superscriptaddress]{revtex4}
\usepackage[utf8]{inputenc}
\usepackage{amsthm,bbm}
\usepackage{amsfonts}
\usepackage{amssymb,enumerate}
\usepackage{amsthm}
\usepackage{wrapfig}
\usepackage{color}
\usepackage{dsfont}
\usepackage{mathtools}
\usepackage{physics}
\usepackage{graphicx}
\usepackage{tikz}
\usepackage{graphics}
\usepackage{sidecap}
\usepackage{caption}
\usepackage{mathtools}

\newcommand{\R}{\mathbb{R}}  
\newcommand{\cF}{\mathcal{F}} 
\newcommand{\cO}{\mathcal{O}} 
\newcommand{\cH}{\mathcal{H}} 
\newcommand{\cB}{\mathcal D(\mathcal{H})} 
\newcommand{\cE}{\mathcal{E}} 
\newcommand{\cL}{\mathcal{L}} 

\def\Tr{\operatorname{Tr}}
\def\>{\rangle}
\def\<{\langle}

\def\mE{\mathcal{E}}
\def\mF{\mathcal{F}}

\def\mL{\mathcal{L}}

   \newcommand{\spa}{\text{span}}

\renewcommand{\qedsymbol}{\nobreak \ifvmode \relax \else
	\ifdim \lastskip<1.5em \hskip-\lastskip \hskip1.5em plus0em
	minus0.5em \fi \nobreak \vrule height0.75em width0.5em
	depth0.25em\fi}

\renewcommand{\geq}{\geqslant}
\renewcommand{\leq}{\leqslant}

\newtheorem{theorem}{Theorem}
\newtheorem{corollary}{Corollary}
\newtheorem{lemma}{Lemma}

\newtheorem{definition}{Definition}
\theoremstyle{remark}

\theoremstyle{definition}

\newcommand{\bea}{\begin{eqnarray}}
\newcommand{\eea}{\end{eqnarray}}
\newcommand{\be}{\begin{equation}}
\newcommand{\ee}{\end{equation}}
\newcommand{\ba}{\begin{equation}\begin{aligned}}
\newcommand{\ea}{\end{aligned}\end{equation}}

\def\be{\begin{equation}}
\def\ee{\end{equation}}
\def\ba{\begin{align}}
\def\ea{\end{align}}

\newcommand{\mH}{\mathcal{H}}

\newcommand{\mK}{\mathcal{K}}

\newcommand{\lr}{\rangle\langle}

\newcommand{\ra}{\rangle}

\begin{document}
\title{Quantifying the Imaginarity of Quantum Mechanics}
\author{Alexander Hickey}
\author{Gilad Gour}
\affiliation{Department of Mathematics and Statistics, University of Calgary, 2500 University Drive NW, Calgary, Alberta, Canada T2N 1N4}
\date{\today}

\begin{abstract}
The use of imaginary numbers in modelling quantum mechanical systems encompasses the wave-like nature of quantum states. Here we introduce a resource theoretic framework for imaginarity, where the free states are taken to be those with density matrices that are real with respect to a fixed basis. This theory is closely related to the resource theory of coherence, as it is basis dependent, and  the imaginary numbers appear in the off-diagonal elements of the density matrix. Unlike coherence however, the set of physically realizable free operations is identical to both completely resource non-generating operations, and stochastically resource non-generating operations. Moreover, the resource theory of imaginarity does not have a self-adjoint resource destroying map. After introducing and characterizing the free operations, we provide several measures of imaginarity, and give necessary and sufficient conditions for pure state transformations via physically consistent free operations in the single shot regime. 
\end{abstract}

\maketitle

\section{Introduction}

Quantum Resource Theories (QRTs) revolutionize the way think about familiar topics like entanglement, coherence, and symmetry. They provide an operational way to quantify and manipulate a given property (resource) of a quantum mechanical system,  and determine the information processing tasks that are attainable by consuming this resource. A well known example of a resource theory is entanglement theory \cite{Horodecki2009}, where the consumption of entanglement leads to protocols such as superdense coding \cite{Bennett1992} and quantum teleportation \cite{Bennett1993}. In recent years, there has been interest in developing resource theories for many other non-classical properties, such as resource theories of coherence \cite{Baumgratz2014,Chitambar2016, ChitambarHsieh2016, Winter2016, Napoli2016,Marvian2016}, athermality \cite{Brandao2013, Horodecki2013, Lostaglio2015, Gour2014, Narasimhachar2015}, asymmetry \cite{Gour2008,Gour2009, Skotiniotis2012, Marvian2013}, knowledge \cite{Rio2015}, magic \cite{Veitch2014, Howard2017, Ahmadi2017}, steering \cite{Gallego2015}, nonlocality \cite{Vicente2014}, contextuality \cite{Amaral2017}, and superposition \cite{Theurer2017}. Furthermore, there has been interest in developing mathematical frameworks that apply to large classes of resource theories \cite{Brandao2015,Coecke2016,Gour2017,Fritz2017}.

In general, a resource theory is completely characterized by physical restrictions on the set of possible transformations \cite{Gour2017,Coecke2016,Fritz2017}. These restrictions give rise to a set of \textit{free states} $\cF$, and a set of \textit{free operations}, $\cO$, which do not generate a resource. It is therefore a minimal requirement that the free operations acts invariantly on $\cF$. Often times the physical restrictions  lead to a much smaller set of free operations than just those that are resource non-generating (RNG). For example, Local Operations and Classical Communication (LOCC) in the resource theory of entanglement, do not contain all operations that do not generate entanglement when acting on separable states. The physical restriction to LOCC  leads to a much smaller class of operations than the full set of non-entangling operations. Yet, due to the notorious complexity of LOCC operations, other subsets of non-entangling operations are found to be useful, especially those which have a nice mathematical structure, such as PPT operations and separable operations (SEP) in entanglement theory \cite{Chitambar2016}.

A key feature of quantum mechanics is the necessity of imaginary numbers to accurately model the dynamics of a physical system. Although imaginary numbers have long been used in classical physics to simplify models of oscillatory motion and wave mechanics, it seems that they play a much deeper role in quantum physics as they are intrinsic to any orthodox formulation \cite{Wootters2012,Hardy2012,Aleksandrova2013}. Consider, for example, the polarization density matrix of a single photon in the $\left \{ \ket H, \ket V \right \}$ basis. The imaginary numbers in the density matrix gives rise to rotations of the electric field vector, i.e. elliptical or circular polarization. Motivated by this unique characteristic of quantum theory, we propose here a resource theory of imaginarity in which the resource arises from imaginary terms that appear in the density matrix of a state. With such a resource at hand, \emph{real} quantum mechanics can achieve the same tasks as standard quantum mechanics.
 
Such a resource theory is interesting for several reasons. The resource of imaginarity is closely related to quantum coherence as the imaginary terms of a density matrix always appear as off-diagonal elements. This relaxation of only considering the imaginary part of the off-diagonal elements provides a characterization of how relative phases between the measurement basis states effect the underlying dynamics of a system, as opposed to the total interference between basis states as is the case with coherence. From a conceptual standpoint, this theory may help to quantify the seemingly fundamental role that imaginary numbers play in quantum theory. In addition, a characterization of this theory is of interest as it is a part of a large class of QRTs known as affine resource theories, which also includes coherence theory. Recently, work has been done in developing theorems which can be applied to a general affine resource theory, and currently the resource theory of imaginarity is the only known affine resource theory which does not have a self-adjoint resource destroying map \cite{Gour2017}.

This paper is structured as follows. First we formally define free states and completely characterize the largest possible class of free operations. Next we discuss notions of physical consistency and show that by further restricting our free operations, we obtain a class of operations which simultaneously amends multiple physical inconsistencies (something that does not happen in the resource theory of coherence). We will then introduce the notion of a measure of imaginarity and propose several measures. Finally, we will discuss pure state transformations in the single-shot regime via free operations, where we will give necessary and sufficient conditions for the existence of such a transformation. \\

\section{Free States}

Analogous to coherence, imaginarity is intrinsically a basis dependent phenomena. Before defining our free states, we must therefore fix an orthonormal reference basis $\{ \ket{j} \}_{j=0}^{d-1}$ for our $d$-dimensional Hilbert space $\cH$. Such a requirement is not a detriment to the theory however, as in general, one fixes an experimental setup corresponding to some projective basis measurement prescribed by the physics of the system of interest. In the case of a composite state space $\cH_A \otimes \cH_B$, we reference the canonical choice of basis $\left\{ \ket{j}_A \otimes \ket{k}_B \right \}$ where the $k$ index is iterated first. Throughout this paper we will use $\cH$ to denote a $d$-dimensional Hilbert space, $\cL (\cH)$ to denote the set of linear operators on $\cH$ and $\cB$ to denote the set of density operators acting on $\cH$, which we can represent as matrices with respect to the fixed reference basis.

\begin{definition}
Let $\rho$ be a density operator on a $d$-dimensional Hilbert space $\cH$. We say that $\rho$ is real or free if $$\rho = \sum_{jk} \rho_{jk} \dyad{j}{k}$$
where each $\rho_{jk} \in \R$. We denote the set of all real density operators by $\cF$.
\end{definition} 

Note that since density matrices are Hermitian, conjugation is equivalent to transposition. It therefore follows that $\rho \in \cF$ if and only if $\rho$ is a symmetric matrix, that is, $\rho^T=\rho$. It is clear that $\cF$ is convex, as any convex combination of real density operators will remain a real density operator.
Moreover, $\mF$ is also affine~\cite{Gour2017}, that is, any density operator that can be expressed as an affine combination of real density operators is itself real.

\section{Free Operations}
In some resource theories, the set of free operations does not correspond to physical or experimental restrictions on the set of possible operations. Instead, one is interested in studying a certain property of a physical system and uses a resource theoretic framework to do so. For example, the resource theory of coherence is a state-centred theory, in which the set of free states (incoherent states) is well defined whereas the choice of free operations is not unique. In fact, there is some controversy as to which set of free (incoherent) operations is most physical~\cite{Chitambar2016,Marvian2016}.

The minimal condition on any set of free operations is that a free operation can only map a free state to a free state. This ensures that the set of free operations do not generate resources from free states. The largest possible set of free operations in any conceivable resource theory is therefore the resource non-generating (RNG) operations. As we will discuss below however, RNG operations are typically too large in the sense that they are not physically consistent. Yet in order to derive a more physically consistent QRT of imaginarity, we start with the following characterization of RNG operations.

\begin{theorem} \label{thm:MRO}
Let $\cE:\cL(\mH_B)\to\cL(\mH_A)$ be a quantum channel and $$J = \sum_{jk} \cE (\dyad{j}{k}) \otimes \dyad{j}{k}$$ the Choi matrix of $\cE$ acting on $\cH_A \otimes \cH_B$. Then $\cE$ is resource non-generating if and only if $J-J^{\Gamma_A}$ is symmetric. Here $\Gamma_A$ denotes the partial transpose with respect to the Hilbert space $\cH_A$.
\end{theorem}

\begin{proof}
Suppose $\cE$ is RNG, then for any $\rho \in \cF$ we have that $[\cE(\rho)]^T=\cE(\rho)$. Rewriting this condition in the Choi representation we get that  
\begin{align*}
\Tr_B \left(  \left[ J(I\otimes \rho^T) \right]^{\Gamma_A}  \right) &= \Tr_B \left( J^{\Gamma_A} (I \otimes \rho^T)\right) \\ &= \Tr_B \left( J(I\otimes \rho^T) \right).
\end{align*} So we have $$\Tr_B \left[ (J-J^{\Gamma_A})(I \otimes \rho^T) \right] =0$$ and thus $$\Tr \left[ (J-J^{\Gamma_A})(X \otimes \rho) \right]=0$$ for any $X \in \cL(\cH_A)$ and symmetric $\rho$. We may view this as an orthogonality condition with respect to the Hilbert-Schmidt inner product, and conclude that
\begin{equation*}
J-J^{\Gamma_A}\in\mL(\mH_A)\otimes\mK_B\;,
\end{equation*}
where $\mK_B$ is the space of antisymmetric operators in $\mL(\mH_B)$. This in particular implies that
$$
\left( J-J^{\Gamma_A} \right)^{\Gamma_B} = J^{\Gamma_A}-J \;,
$$ 
so that 
 \begin{align*}
 \left( J-J^{\Gamma_A} \right)^T &= \left( \left[ J-J^{\Gamma_A} \right]^{\Gamma_B} \right)^{\Gamma_A} \\
&= \left( J^{\Gamma_A}-J \right)^{\Gamma_A} \\
&= J-J^{\Gamma_A}.
 \end{align*}
 For the converse, suppose $J-J^{\Gamma_A}$ is symmetric. Then 
 $$
 J - J^{\Gamma_A}=\left(J-J^{\Gamma_A} \right)^T =  J^T - J^{\Gamma_B}
 $$ 
 and so 
 $$ \left( J-J^{\Gamma_A} \right)^{\Gamma_B}= (J^T - J^{\Gamma_B})^{\Gamma_B}= -\left( J-J^{\Gamma_A} \right) .$$ Thus
\begin{align*}
\cE(\rho) &= \Tr_B \left[ J(I\otimes \rho^T) \right] \\
&= \Tr_B \left[ J^{\Gamma_A}(I\otimes \rho^T) \right]+ \Tr_B \left[ \left( J-J^{\Gamma_A} \right)(I\otimes \rho^T) \right] \\
&= [\cE(\rho)]^T + \Tr_B \left[ \left( J-J^{\Gamma_A} \right)(I\otimes \rho^T) \right].
\end{align*}
Finally, note that
\begin{align*}
\Tr_B \left[ \left( J-J^{\Gamma_A} \right)(I\otimes \rho^T) \right] &= \Tr_B \left[ \left( J-J^{\Gamma_A} \right)^{\Gamma_B}(I\otimes \rho^T)^{\Gamma_B} \right] \\
&= -\Tr_B \left[ \left( J-J^{\Gamma_A} \right)(I\otimes \rho) \right]
\end{align*}
and so $\Tr_B \left[ \left( J-J^{\Gamma_A} \right)(I\otimes \rho^T) \right]= 0$ whenever $\rho$ is symmetric, therefore $\cE(\rho) = [\cE(\rho)]^T$ whenever $\rho \in \cF$.
\end{proof}

\subsection{Physical Constraints on the Free Operations}

The classification of the complete set of RNG operations is of interest as it was shown in \cite{Brandao2015} that a QRT is asymptotically reversible if its set of free operations is maximal. We will however, concern ourselves with a restricted class of free operations based on the ground of physical consistency. In what follows, we define three types of such free operations and then show that they are all identical. We start by showing that there are RNG operations as defined in Theorem~\ref{thm:MRO} which can generate resources when applied to part of a larger composite system. 

Theorem \ref{thm:MRO} classifies an arbitrary RNG operation with the condition that the Choi matrix must satisfy a symmetry constraint upon taking a partial transpose. It is straightforward to construct (even in the qubit case) a Choi matrix $J$ which satisfies the condition in Theorem~\ref{thm:MRO} but is not  itself symmetric. On the other hand,
we have $\frac 1d J= \cE \otimes \text{id}_d \left(\dyad{\phi^+} \right)$ where $\ket{\phi^+}= \frac{1}{\sqrt d}\sum_j \ket{jj} $ is the normalized maximally entangled state. It therefore follows that $\dyad{\phi^+} \in \cF$ but 
$$
\cE \otimes \text{id}_d \left(\dyad{\phi^+} \right) = \frac 1d J \not\in \cF\;;
$$
i.e.  $\cE$ is a RNG operation that generates a resource when applied to part of a maximally entangled state. This inconsistency can be rectified by requiring the free operations to be closed under tensor products, a property which we will call \emph{completely resource non-generating}.

\begin{definition}
A RNG operation, $\mE$, is said to be \textit{completely resource non-generating} if $\cE \otimes \text{id}_k$ is resource non-generating for all $k \in \mathbb N$.
\end{definition}

This definition is analogous to the class of completely non-entangling operations introduced in \cite{Chitambar2017}, and ensures that resources can never be generated by applying a free operation to part of a larger composite system. 

One may also impose the requirement that resources cannot be generated stochastically, i.e. with probability less than one. This leads to an alternate class of free operations which we call stochastically resource non-generating operations.

\begin{definition}
A quantum operation $\cE$  is said to be \emph{stochastically RNG} if it has an operator sum representation with Kraus operators $\{ K_j\}_{j=1}^{m}$ such that
$$\frac{K_j \rho K_j^\dagger}{\Tr(K_j \rho K_j^\dagger)} \in \cF$$ for all $j$ whenever $\rho \in \cF$.
\end{definition}

This class of free operations ensures that resources cannot be generated probabilistically, and is precisely how incoherent operations (IO) are defined in the resource theory of coherence \cite{Baumgratz2014}.

Furthermore, it is argued in \cite{Chitambar2016, Marvian2016} that for a resource theory to be physically realizable, the set of free operations should admit a free dilation. This meaning that the operation can be implemented via free unitaries acting on some larger composite system.

\begin{definition}
A unitary $U$ is said to be a \textit{free unitary} if $U \rho U^\dagger \in \cF$ whenever $\rho \in \cF$.
\end{definition}

\begin{definition} \label{defn:PRO}
A RNG operation $\cE$ is said to be \textit{physically realizable} if it admits a free dilation. That is, there exists a Hilbert space $\cH_E$, a free state $\dyad{0}_E \in \cH_E$, and a free unitary $U_{AE}$ acting on the joint state space $\cH_A \otimes \cH_E$ such that $$\cE (\rho)= \Tr_E \Big[ U_{AE} \left( \rho \otimes \dyad{0}_E \right) U_{AE}^\dagger \Big]$$ for any density operator $\rho \in \cH$.
\end{definition} 

In general, the class of physically realizable operations can be much smaller than the class of all RNG operations (see Fig. \ref{fig:venn1}). However, recall that a CPTP map (and a generalized measurement) provides only an \emph{effective} description of the evolution of a physical system. Basically, the evolution of a closed system is always unitary. Therefore, in order to implement a free CPTP map, one has to first implement a free unitary on the system+ancilla and then trace out the ancillary system. This means that many CPTP maps that are RNG, and even completely RNG, cannot be implemented freely, and in this sense are not really free. 

\begin{figure}[t]
    \centering
    \includegraphics[scale=0.35]{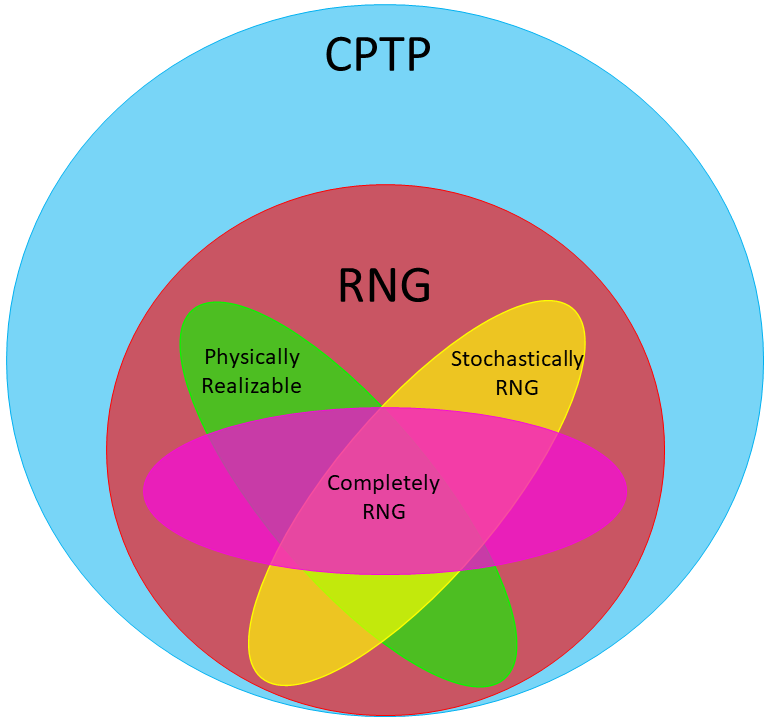}
    \caption{A Venn diagram showing some possible choices of free operations. Resource non-generating (RNG) operations form the largest possible set of free operations within the set of all CPTP maps. Completely RNG operations ensure that the operations are closed under tensor products. Stochastically RNG operations ensure that resources cannot be generated probabilistically. Physically realizable operations are those which can be implemented at no cost.}
    \label{fig:venn1}
\end{figure}

It follows from definition \ref{defn:PRO} that the class of free unitaries must be characterized in order to implement a physically realizable operations. This is done in the following lemma.

\begin{lemma}
A unitary matrix $U$ is a free unitary if and only if $U= e^{i\theta} Q$ for some $\theta \in [0,2\pi)$ and real orthogonal matrix $Q$.
\label{thm:NSCunitary}
\end{lemma}

\begin{proof}
Let $U$ be a free unitary, then for any $\rho \in \cF$ we have that $$ \left[ U\rho U^\dagger \right]^T=U\rho U^\dagger $$
and so 
$$\rho = (U^TU) \rho (U^TU)^\dagger$$
where $U^TU$ is also unitary. This equation must hold for all symmetric $\rho$ so for each $j=0,1,...,d-1$:
$$U^TU \dyad{j}{j} (U^TU)^\dagger = (U^TU \ket{j})(U^TU \ket{j})^\dagger =\dyad{j}{j}$$
and thus $U^TU \ket{j} = e^{i\theta_j} \ket{j}$ for some $\theta_j \in [0,2\pi)$. Similarly, for each $k=1,2,...,d-1$ and $j<k$:
\begin{align*}
\dyad{j}{k}+\dyad{k}{j} &= U^TU(\dyad{j}{k}+\dyad{k}{j})(U^TU)^\dagger \\ 
 &=e^{i(\theta_j-\theta_k)} \dyad{j}{k} +e^{i(\theta_k-\theta_j)} \dyad{k}{j}
\end{align*}
therefore $\theta_j = \theta_k$ for each $j,k$ and thus $\theta_0=\theta_1=...=\theta_{d-1} =: \theta$. It follows that $$U^TU = e^{i\theta} \sum_{j=0}^{d-1} \dyad{j} = e^{i\theta} I.$$ Define $Q := e^{-i\theta/2}U$, then $Q^TQ=I$ and $Q^\dagger Q=I$ so $Q$ is real and orthogonal. The converse is trivial as clearly $[U\rho U^\dagger]^T = Q\rho^T Q^T \in \cF$ whenever $\rho \in \cF$.
\end{proof}

It follows from Lemma \ref{thm:NSCunitary} that the set of free unitaries in the resource theory of imaginarity differs quite significantly from the free unitaries in coherence theory. In particular, the most general unitary which does not generate coherence takes the form $U = \sum_j e^{i\theta_j} \dyad{\pi(j)}{j}$ where $\pi$ is some permutation on the set of indices $\{0,1,...,d-1\}$ \cite{Marvian2016,Chitambar2016}. So while the resource of imaginarity is preserved under any orthogonal basis transformation, coherence is only preserved under permutations of the basis elements. In this sense, the resource theory of imaginarity contrasts from coherence theory in the sense that some of the basis dependence has been alleviated.

It turns out that each of the possible physical constraints that were discussed corresponds to the same set of free operations, forming an equivalence class of operations that we call physically consistent. We present this fact in the main result of this section.

\begin{theorem} \label{thm:RealChoi}
Let $\cE$ be a quantum channel, then the following are equivalent:
\begin{enumerate}
\item $\cE$ is completely resource non-generating.
\item The Choi matrix of $\cE$ is real.
\item $\cE$ is stochastically resource non-generating.
\item $\cE $ is physically realizable.
\end{enumerate}
\end{theorem}

\begin{figure}[t]
    \centering
    \includegraphics[scale=0.35]{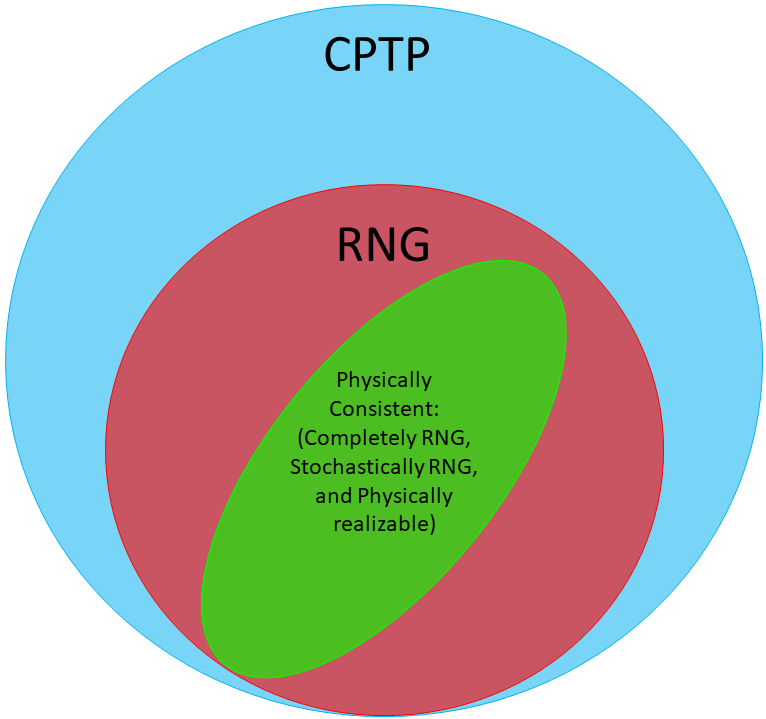}
    \caption{Each of the physical constraints leads to an equivalent class of free operations in the resource theory of imaginarity.}
    \label{fig:venn2}
\end{figure}

\begin{proof}

For $(1) \implies (2)$, suppose $\cE \otimes \text{id}_k$ is RNG for all $k \in \mathbb N$. Fix $k=d$  and let $\ket{\phi^+}= \frac{1}{\sqrt{d}} \sum_j \ket{jj}$, then $\dyad{\phi^+} \in \cF$ and so $$\cE \otimes \text{id}_d(\dyad{\phi^+}{\phi^+}) \in \cF .$$ Thus the Choi matrix is real.

For $(2) \implies (3)$, suppose that the Choi matrix $J$ is real, then $J$ has a spectral decomposition $$J= \sum_{j} \lambda_j \dyad{v_j}{v_j}$$ where $\{ \ket{v_j}\}$ is a set of real orthonormal eigenvectors of $J$. A set of Kraus operators $\{ K_j \}$ can then be obtained by inverting the vectorization operation on each of the eigenvectors $$\text{vec}(K_j)= \sqrt{\lambda_j} \ket{v_j} \hspace{1.5cm} \text{\cite{Watrous2018}} $$
where the vectorization operator is an isomorphism defined in the usual way: $\text{vec}(\dyad{i}{j}) := \ket{j}\otimes \ket{i}$. Since $J$ is positive, each $\sqrt{\lambda_j}$ is real and thus $\{ K_j\}$ is a set of real Kraus operators.

For $(3) \implies (4)$, suppose $\{ K_j \}_{j=1}^N$ is a real set of Kraus operators, then we may construct a unitary $U_{AE}$ acting on the state space along with some ancillary system (just as in definition \ref{defn:PRO}) where $U_{AE}$ satisfies $$K_j= \bra{j}_E U_{AE} \ket{0}_E$$ for each $j$. Since each $K_j$ is real, we may restrict the column space of $U_{AE}$ to $\R^N$ and pick the remaining columns of the matrix to be orthogonal. The matrix $U_{AE}$ will then be real and orthogonal, and  
\begin{align*}
\cE(\rho) &= \sum_j K_j \rho K_j^\dagger \\
&= \sum_j \bra{j}_E U_{AE} \Big[ \ket{0}_E \rho \bra{0}_E \Big] U_{AE}^\dagger  \ket{j}_E \\
&=\sum_j \bra{j}_E U_{AE} \Big[ \rho \otimes \dyad{0}_E \Big] U_{AE}^\dagger  \ket{j}_E \\
&= \Tr_E \Big[ U_{AE} (\rho \otimes \dyad{0}_E) U_{AE}^\dagger \Big]
\end{align*}
Therefore $\cE$ admits a free dilation and thus $\cE \in \cO$.

For $(4) \implies (1)$, suppose $$\cE(\rho) = \Tr_E \Big[ U_{AE} (\rho \otimes \dyad{0}_E) U_{AE}^\dagger \Big]$$ for some real orthogonal matrix $U_{AE}$ and consider a composite state space $\cH \otimes \cH_B$ where $\text{dim} \cH_B = k$ and some joint free state $\sigma \in \mathcal D(\cH \otimes \cH_B) \cap \cF$. Then $\sigma = \sum_{j\ell} r_{j\ell} A_j \otimes B_\ell$ for some $r_{j\ell} \in \mathbb R$ and real matrices $A_j$ and $B_\ell$. Thus
\begin{align*}
\cE \otimes \text{id}_k (\sigma) &= \cE \otimes \text{id}_k \left( \sum_{j\ell} r_{j\ell} A_j \otimes B_\ell \right) \\
&= \sum_{j\ell} r_{j\ell} \cE(A_j) \otimes B_\ell \\
&= \sum_{j\ell} r_{j\ell} \Tr_E \Big[ U_{AE} (A_j \otimes \dyad{0}_E) U_{AE}^\dagger \Big] \otimes B_\ell
\end{align*}
where $U_{AE}$ is a real matrix. Therefore $\cE \otimes \text{id}_k$ is resource non-generating for any $k \in \mathbb N$.
\end{proof}

Taking our set of free operations to be those with real Choi matrices fixes each of the previously discussed physical inconsistencies simultaneously. This is in contrast to coherence theory, where the physically realizable operations (PIO) are a strict subset of those which are stochastically resource non-generating (IO) \cite{Chitambar2016}. We note that the various sets of RNG operations listed in Theorem \ref{thm:RealChoi} are analogous to the usual representations of quantum channels. In particular, we see that if the set of all CPTP maps are restricted to just the completely RNG operations, that these operations admit a Stinespring dilation, as well as operator-sum and Choi representations which are valid when one restricts the underlying Hilbert space to the real numbers. The set of physically consistent free operations therefore provide a characterization of channels in real quantum mechanics which largely resemble the standard theory of CPTP maps.

It is clear that the set of physically consistent operations is convex. Furthermore, this class of operations admit a nice algebraic structure, which we call \textit{transposition covariance}.

\begin{corollary}
A quantum channel $\cE$ is physically consistent if and only if it commutes with the transposition map. That is, $$\cE(\rho)^T = \cE(\rho^T)$$ for all $\rho \in \cB$.
\end{corollary}

\begin{proof}
Suppose $\cE$ is physically consistent, then by theorem \ref{thm:RealChoi}, the Choi matrix $J$ of $\cE$ is real and so $J^T = J$. Then 
\begin{align*}
\cE(\rho)^T &= \Tr_B \left[ J^{\Gamma_A} (I \otimes \rho^T) \right] \\
&=  \Tr_B \left[ J^T (I \otimes \rho) \right] \\
&=  \Tr_B \left[ J (I \otimes \rho) \right] \\
&=  \cE(\rho^T)
\end{align*}
where the second equality comes from the fact that the partial transpose is self-adjoint under the trace. For the converse, note that $$\cE(\rho)^T = \Tr_B[J^{\Gamma_A}(I \otimes \rho^T) ] = \Tr_B[J^T(I \otimes \rho) ]$$ and so 
$$0= \cE(\rho^T)- \cE(\rho)^T = \Tr_B \left[ (J-J^T)(I \otimes \rho) \right]  \quad \forall \rho \in \cB,$$
so that $J=J^T$. 
\end{proof}

The definition of Transposition Covariant Operations is analogous to the set of Dephasing-Covariant Incoherent Operations (DIO) in coherence theory \cite{Chitambar2016}, as both the transposition and completely dephasing maps act as the identity on the set of real and incoherent states respectively. Motivated by the physical consistency of this versatile class of operations, we henceforth restrict our definition of free operations to those which are physically consistent, as characterized in theorem \ref{thm:RealChoi}.

\section{Measures of Imaginarity}

Next we discuss the quantification of imaginarity, which is necessary in determining the degree of resourcefulness of a given state. This will act as a first step in understanding the extent to which a resource may be used to simulate non-free operations. Such an understanding is of particular interest in the resource theory of imaginarity, as it will allow us to gauge how well real quantum mechanics approximates the standard theory.

\begin{definition} \label{defn:measure}
A \textit{measure of imaginarity} is a function $M: \cB \to [0,\infty)$ such that 
\begin{enumerate}
\item $M(\rho) = 0$ if $\rho \in \cF$
\item $M(\cE(\rho)) \leq M(\rho)$ whenever $\cE$ is a free operation. 
\end{enumerate}
\end{definition}
Furthermore, we say that a measure of Imaginarity $M$ is \textit{faithful} if $M(\rho) = 0 \implies \rho \in \cF$, and that $M$ is an \textit{imaginarity monotone} if it does not increase on average under free operations. Condition 1 requires that any real state will have no imaginarity. Condition 2 follows from the initial assumption that our free operations should not increase the resource content of a state. In general, any measure of a resource should be monotonically non-increasing under the physically free operations, but if the measure is monotonic under RNG, then the measure remains valid for any resource theory with the same set of free states \cite{Brandao2015}.

As discussed in \cite{Brandao2015}, there are many choices of valid measures for an arbitrary convex resource theory, which are often given by minimizing some quantity over a subset of linear operators. A couple of examples that will satisfy definition \ref{defn:measure} include:
\begin{enumerate}[(i)]
\item The robustness of imaginarity $$\mathcal R(\rho) = \min_{\pi \in \cB} \left\{ s \geq 0 : \frac{s \pi + \rho}{1+s}\in \cF \right\}.$$
\item A distance measure $$M(\rho) = \min_{\sigma \in \cF} \mathcal C(\rho,\sigma)$$ where $\mathcal C$ is any contractive metric.
\end{enumerate}

\begin{figure}[t]
\centering
	\begin{tikzpicture}
	\node[inner sep=0pt] at (0,0)
	    {\includegraphics[width=.4\textwidth]{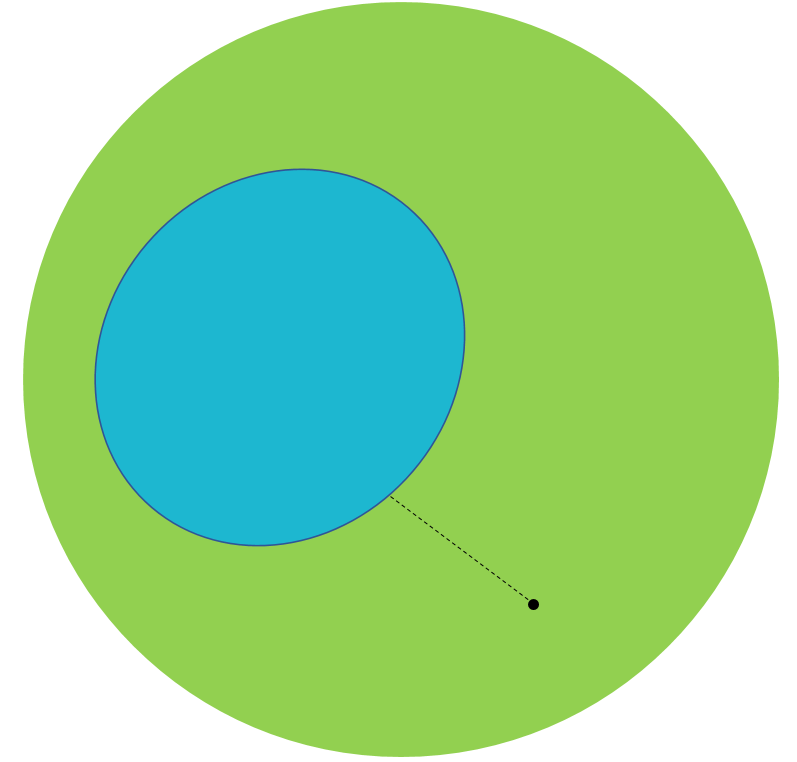}};
	\draw (-1.2,0.4) node[below] {\Huge $\cF$};
	\draw (1.9,1.2) node[below] {\Huge $\cB$};
	\draw (1.4,-1.9) node[below] {\large $\rho$};
	\draw (0.9,-1.1) node[below,rotate=-30] {\normalsize $M(\rho)$};
	\end{tikzpicture}
	\caption{The measure introduced in theorem \ref{thm:measure} is interpreted as the trace distance to the set of real states.}
	\label{fig:venn1}
\end{figure}
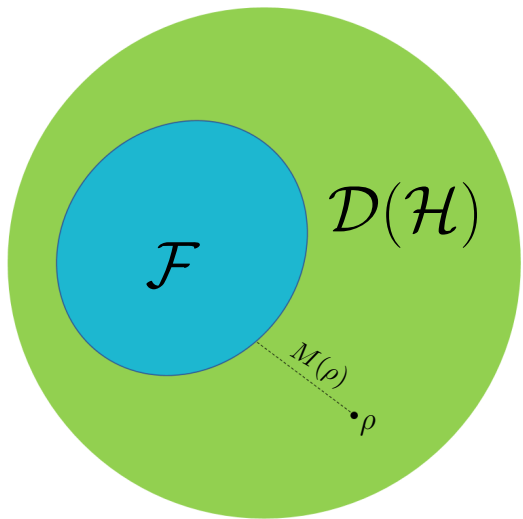

We will focus our attention to a special case of example (ii), where we consider the contractive metric to be the trace distance.

\begin{theorem} \label{thm:measure}
Let $\rho$ be a density matrix, then $$M(\rho) := \min_{\sigma \in F} \norm{\rho-\sigma}_1 = \frac 12 \norm{\rho-\rho^T}_1$$ is a faithful measure of imaginarity.
\end{theorem}

\begin{proof}
We will show that the trace distance is minimized by choosing $\sigma = \frac 12 \left( \rho + \rho^T \right)$ by showing that $\frac 12\norm{\rho-\rho^T}_1$ satisfies Definition \ref{defn:measure}. It is clear that $M(\rho)=0$ if and only if $\rho=\rho^T$ since the trace norm is a norm on $\cL(\cH)$. To show the monotonicity condition, recall that the trace norm is contractive under CPTP maps, so $$\norm{\cE(\rho) - \cE(\sigma)}_1 \leq \norm{\rho - \sigma}_1$$ for any CPTP map $\cE$ and $\rho,\sigma \in \cB$. Since $\rho$ is Hermitian, we can decompose $\rho = \rho_R + i\rho_I$ where $\rho_R = \frac 12 (\rho+\rho^T)$ is real and symmetric and $\rho_I = \frac{1}{2i}(\rho - \rho^T)$ is real and antisymmetric. Next note that $$\Tr \rho_R = \Tr \rho = 1$$ and $$\bra{\psi}\rho_R \ket{\psi} = \frac 12 \bra{\psi}\rho \ket{\psi} + \frac 12 \bra{\psi}\rho^T \ket{\psi} \geq 0$$ for any $\ket{\psi} \in \cH$, therefore $\rho_R \in \cB \cap \cF$. We can therefore pick $\sigma = \rho_R$ to get $$\norm{i\cE(\rho_I)}_1 \leq \norm{i \rho_I}_1.$$ Finally, note that if $\cE$ is a free operation then $$\cE(\rho_I) = \frac{1}{2i} \left[ \cE(\rho)- \cE(\rho)^T \right] = \frac{1}{2i} \left[ \cE(\rho)- \cE(\rho^T) \right] $$ and so $\norm{\cE(\rho_I)}_1 = M(\cE(\rho))$. Therefore $$M(\cE(\rho)) \leq M(\rho)$$
\end{proof}

As an example, consider the case where $\rho$ is a qubit density matrix. Then $\rho$ can be expanded in the Pauli basis $$\rho = \frac 12 (I+x\sigma_x+y\sigma_y+z\sigma_z)$$ where $x,y,z \in \R$ are the components of the Bloch vector of $\rho$ \cite{Nielsen2002}. Then $\frac 12(\rho-\rho^T) = y \sigma_y$ and thus $$M(\rho) = \frac 12 \norm{\rho-\rho^T}_1 =  \norm{y \sigma_y}_1 = 2|y|= 2 |\Tr(\rho \sigma_y)|$$
since the eigenvalues of $\sigma_y$ are $\pm 1$. Therefore in the qubit case, the measure defined in theorem \ref{thm:measure} reduces to twice the $y-$component of the Bloch vector.

\section{State Transformations}

Next we will consider the problem of transforming resourceful states via free operations. That is, given two resourceful states $\rho$ and $\rho'$, is there a free operation $\cE \in \cO$ such that $\cE ( \rho )= \rho'$? This is an important question in any resource theory as it induces a hierarchy of resourcefulness to be used in information processing tasks. In particular, we will consider transformations between single copies of a pure state. To answer this question, it will be useful first to determine a standard way to represent the resource of a given state. By manipulating resourceful states with free unitary operations, one can ensure that the resource content does not change as the transformation is reversible and its inverse is also free. The following lemma presents a standardized form for any pure state.

\begin{lemma} \label{thm:std_form}
Let $\ket{\psi}$ be a pure state in a $d-$dimensional Hilbert space with reference basis $\{ \ket{j} \}_{j=0}^{d-1}$. Then there exists a free unitary $U$ such that $U|\psi\ra\in\spa\{|0\ra,|1\ra\}$.
\end{lemma}

\begin{proof}
Decompose $|\psi\ra=a|\psi_R\ra+ib|\psi_I\ra$, with $|\psi_R\ra$ and $|\psi_I\ra$ being two real normalized vectors, and $a,b$ two positive real numbers satisfying $a^2+b^2=1$. Let $O$ be an orthogonal real matrix satisfying $O|\psi_R\ra=|0\ra$. It follows that $|\psi\ra$ is resource equivalent to $O|\psi\ra= a|0\ra+ib|\phi\ra$, where $|\phi\ra\equiv O|\psi_I\ra$ is some real normalized vector. Next, denote $|\phi\ra=\cos(\theta)|0\ra+\sin(\theta)|\chi\ra$, where $|\chi\ra$ is a real vector in the span of $\{|j\ra\}_{j=1}^{d-1}$, and
apply another orthogonal matrix  in the form $O'\equiv |0\lr 0|\oplus\tilde{O}$, where $\tilde{O}$ is an orthogonal matrix in the $d-1$-dimensional space spanned by $\{|j\ra\}_{j=1}^{d-1}$, satisfying $\tilde{O}|\chi\ra=|1\ra$.
We therefore get that
$$
O'O|\psi\ra= \left(a+ib\cos(\theta)\right)|0\ra+ib\sin(\theta)|1\ra\;.
$$
Hence, any state $|\psi\ra$ can be converted reversibly by free operations into a pure qubit state in the span of $\{|0\ra,|1\ra\}$.
\end{proof}

We can thus represent the resource of an arbitrary pure state with a vector in the two dimensional subspace spanned by the first two reference basis vectors. However, there is still more that we can do:
\begin{lemma}
Let $\rho$ be a qubit \emph{mixed} state. Then there exists a $2\times 2$ orthogonal matrix $O$ such that
$$ O\rho O^T=\begin{pmatrix}
1/2 & x-iy\\
x+iy & 1/2
\end{pmatrix}\;, $$
where $x$ and $y$  are non-negative real numbers such that $x^2+y^2\leq 1/4$.
\end{lemma}

\begin{proof}
Denote $$\rho=\begin{pmatrix}
a & c\\
\bar{c} & 1-a
\end{pmatrix}
$$ with $0\leq a\leq 1$ and $c$ a complex number with $|c|^2\leq a(1-a)$. Take 
$$
O=\begin{pmatrix}
\cos\alpha & \sin\alpha\\
-\sin\alpha & \cos\alpha
\end{pmatrix} 
\quad\text{
with} \quad\tan(2\alpha)=\frac{2a-1}{c+\bar{c}}\;.$$ Then, it is straight forward to check that $O\rho O^T$ has 1/2 along the diagonal.
\end{proof}

Note that if $\rho$ is a pure state then $x^2+y^2=1/4$. We therefore conclude that any $d$-dimensional pure state $|\psi\ra$ can be converted reversibly by free operations to a state of the form
$$ |\theta\ra\equiv\frac{1}{\sqrt{2}}\left(|0\ra+e^{i\theta}|1\ra\right)\quad\text{where}\quad x+iy=\frac{1}{2}e^{i\theta}\;. $$
Note that since $x,y\geq 0$ we can assume w.l.o.g. that $0\leq \theta\leq\frac{\pi}{2}$.
Hence, the resourcefulness of pure states is determined completely by one parameter $\theta$. Moreover, note that
we have
$$
M(|\psi\lr\psi|)=M(|\theta\lr\theta|)=2\sin(\theta)\;.
$$

\begin{theorem} \label{thm:PROtransform}
Let $|\psi\ra$ and $|\phi\ra$ be two pure states in a $d$-dimensional Hilbert space $\mH$. Then there exists a free operation $\cE$ such that $\cE(\dyad{\psi})=\dyad{\phi}$ if and only if 
$$
M(|\psi\lr\psi|)\geq M(|\phi\lr\phi|)\;.
$$
\end{theorem}

\begin{proof}
The necessity of this condition follows from the fact that the measure $M$ is a monotone.  It is therefore left to show that if $M(|\psi\lr\psi|)\geq M(|\phi\lr\phi|)$ then there exists a free operation $\mE$ such that $\cE(\dyad{\psi})=\dyad{\phi}$. 

Recall that every qubit state $\rho$ can be expanded in the Pauli basis $$\rho = \frac 12(I + \vec{\mathbf{r}} \cdot \vec{\mathbf{\sigma}})$$ where $\vec{\mathbf{r}}: \, \norm{\vec{\mathbf{r}}} \leq 1$ is a real 3-component vector called the Bloch vector of $\rho$ and $\vec{\mathbf{\sigma}} = (\mathbf{X} \, \mathbf{Y} \, \mathbf{Z})^T$ is a vector of Pauli matrices. Next, every quantum operation $\cE$ on a qubit has a geometric interpretation as an affine transformation which maps the Bloch sphere onto itself \cite{Pasieka2009}. The action of the operation $\cE$ on $\rho$ will thus produce a new state $\rho'$ with corresponding Bloch vector given by 
\be \label{eqn:bloch_transform}
\vec{\mathbf{r}}' = T \vec{\mathbf{r}} + \vec{\mathbf{t}}
\ee
where $\vec{\mathbf{t}}$ is a 3-component real vector and $T$ is a real 3x3 matrix. In our case, w.l.o.g. we assume that the input state $\rho$ is the pure state $|\psi\ra=|\theta\ra$ and the output state $|\phi\ra$ is the pure state $|\theta'\ra$ with $0\leq \theta'\leq\theta\leq 1$. Consequently, we have
$$\vec{\mathbf{r}}=\begin{pmatrix}
\cos(\theta)\\
\sin(\theta)\\
0
\end{pmatrix}\quad
\text{and}\quad\vec{\mathbf{r}}'=\begin{pmatrix}
\cos(\theta')\\
\sin(\theta')\\
0
\end{pmatrix}$$
Clearly, Eq.~\eqref{eqn:bloch_transform} holds for
$$
T= \begin{pmatrix}
1 & 0 & 0 \\
0 & \frac{\sin(\theta')}{\sin(\theta)} & 0 \\
0 & 0 & 0 
\end{pmatrix}\quad\text{and}\quad\vec{\mathbf{t}} = \begin{pmatrix}
\cos(\theta')-\cos(\theta) \\
0 \\
0
\end{pmatrix} 
$$
Moreover, note that $\frac{\sin(\theta')}{\sin(\theta)}\leq 1$ and $\cos(\theta')-\cos(\theta)\geq 0$. It is straight forward to check that these choices lead to a Choi matrix that is real and positive semi-definite.
\end{proof}

It follows from the theorem above, that the state $\ket{+ i}:= \frac{1}{\sqrt 2}\left( \ket{0} + i \ket{1}\right)$,  can be transformed to all other $d$-dimensional pure states via free operations. This, in turn, implies that $\ket{+ i}$ can also be converted to any mixed state, since $\ket{+ i}$ can be converted to any state $|\phi_j\ra$ with probability $p_j$, achieving the transformation $|\psi\ra\to\rho\equiv\sum_jp_j|\phi_j\lr\phi_j|$.
This means that the resource theory of imaginarity has a unique maximally resourceful state, $\ket{+ i}$, that can be converted to any other state, which we call the \textit{maximally imaginary} state. Note that $\ket{-i}$ and $\ket{+i}$ are related by a free orthogonal transformation and therefore do not correspond to two distinct maximally imaginary states. 

The characterization of the transformations attainable by physically consistent free operations in Theorem \ref{thm:PROtransform} allows one to achieve the same tasks as standard quantum mechanics with the maximally imaginary state. For example, one may use the $\ket{+i}$ state to generate the product state $\ket{+i}\ket{+i}$ and subsequently transform a single copy of the maximally imaginary state to any other state deterministically.

The fact that the existence of a process which takes one pure state to another requires the comparison of only a single parameter is not too surprising, as Lemma \ref{thm:std_form} shows that we can equivalently represent the resourcefulness of any pure state by a qubit. This is in contrast to many other resource theories, for example, in entanglement theory with LOCC, one must evaluate $d-1$ inequalities to determine whether or not a free channel exists between two pure states \cite{Nielsen2002}. It is worth noting that a similar characterization exists in the resource theory of the $\mathbb Z_2$-superselection rule \cite{Gour2008}, where one can reversibly map any pure state to a two-dimensional subspace and the condition for the existence of a state transformation also depends only on a single parameter.

\section{Conclusion}
We introduced a framework for the resource theory of imaginarity, which was first mentioned in \cite{Gour2017} where it was shown to be the only known affine resource theory that does not have a self-adjoint resource destroying map. Upon defining real states, we completely characterized two sets of real operations; those which are resource non-generating and those which we call physically consistent. We showed that not only do the physically consistent free operations prevent resources from being generated probabilistically, but they coincide with all completely resource non-generating operations as well as admit a free dilation. This provides a level of consistency which is not present in coherence theory \cite{Chitambar2016}, and further suggests that imaginarity is a physically significant property of a quantum system. We also introduced several measures of imaginarity and provided a closed form measure in terms of the trace distance. Furthermore, Theorem \ref{thm:PROtransform} provides necessary and sufficient conditions for the the existence of a physically consistent channel between two pure states, and shows that there is an equivalence class of states in any dimension from which any pure state can be generated.

The results presented in this paper may act as a starting point for the analysis of other imaginarity measures, mixed state transformations, transformations in the asymptotic limit of many copies of states and the possible existence of catalysts \cite{Jonathan1999,Aberg2014,Duarte2016} under the restriction to free operations. An interesting direction for future work would be to investigate deviations from the complex numbers, in a resource theoretic framework, with reference to number systems other than the reals. This could be a restriction to some smaller field such as the rationals, or an extension of the complex numbers such as the quaternions. The next steps in this theory could include extensions to infinite dimensional Hilbert spaces as well as analyzing the connection of this theory to entanglement theory \cite{Killoran2016} and other convex resource theories \cite{Gour2017}.

{\begin{acknowledgements}

We thank Mark Girard and Nuiok Dicaire for useful discussions and comments. A. Hickey is supported by NSERC under an Undergraduate Student Research Award. GG acknowledges support from NSERC.
\end{acknowledgements}}


\begin{references}

\bibitem{Horodecki2009}
R. Horodecki, P. Horodecki, M. Horodecki, and K. Horodecki, Rev. Mod. Phys. \textbf{81}, 865 (2009).

\bibitem{Bennett1993}
C. H. Bennett, G. Brassard, C. Cr\~apeau, R. Jozsa, A. Peres, and W. K. Wootters, Phys. Rev. Lett. \textbf{70}, 1895 (1993).

\bibitem{Bennett1992}
C. H. Bennett and S. J. Wiesner, Phys. Rev. Lett. \textbf{69}, 2881 (1992).

\bibitem{Baumgratz2014}
T. Baumgratz, M. Cramer, and M. B. Plenio, Phys. Rev. Lett. \textbf{113}, 140401 (2014).

\bibitem{Chitambar2016}
E. Chitambar and G. Gour, Phys. Rev. Lett. \textbf{117}, 030401 (2016).

\bibitem{Marvian2016}
I. Marvian and R. W. Spekkens, Phys. Rev. A \textbf{94}, 052324 (2016).

\bibitem{ChitambarHsieh2016}
E. Chitambar and M.-H. Hsieh, Phys. Rev. Lett. \textbf{117}, 020402 (2016).

\bibitem{Winter2016}
A. Winter and D. Yang, Phys. Rev. Lett. \textbf{116}, 120404 (2016).

\bibitem{Napoli2016}
C. Napoli, T. R. Bromley, M. Cianciaruso, M. Piani, N. Johnston, and Gerardo Adesso, Phys. Rev. Lett. \textbf{116}, 150502 (2016).

\bibitem{Brandao2013}
F. G.S.L. Brand\~ao, M. Horodecki, J. Oppenheim, J. M.
Renes, and R. W. Spekkens, Phys. Rev. Lett. \textbf{111}, 250404 (2013).

\bibitem{Horodecki2013}
M. Horodecki and J. Oppenheim, Nature Communications \textbf{4}, 2059 (2013).

\bibitem{Lostaglio2015}
M. Lostaglio, D. Jennings, and T. Rudolph, Nature Communications \textbf{6}, 6383 (2015).

\bibitem{Gour2014}
G. Gour, M. P. Muller, V. Narasimhachar, R. W. Spekkens, and N. Y. Halpern, arXiv:1309.6586 (2014).

\bibitem{Narasimhachar2015}
V. Narasimhachar and G. Gour, Nature Communications \textbf{6}, 7689 (2015).

\bibitem{Gour2008}
G. Gour and R. W. Spekkens, New J. Phys. 10 (2008) 033023 

\bibitem{Gour2009}
G. Gour, I. Marvian, and R. W. Spekkens, Phys. Rev. A \textbf{80}, 012307 (2009).

\bibitem{Skotiniotis2012}
M. Skotiniotis and G. Gour, New J. Phys. \textbf{14}, 073022 (2012).


\bibitem{Marvian2013}
I. Marvian and R. W. Spekkens, New J. Phys. \textbf{15}, 033001 (2013).

\bibitem{Rio2015}
L. del Rio, L. Kraemer, and R. Renner, arXiv:1511.08818 (2015).

\bibitem{Veitch2014}
V. Veitch, S. A. H. Mousavian, D. Gottesman, and J. Emerson, New J. Phys. \textbf{16}, 013009 (2014).

\bibitem{Howard2017}
M. Howard and E. T. Campbell, Phys. Rev. Lett. \textbf{118}, 090501 (2017).

\bibitem{Ahmadi2017}
M. Ahmadi, H. B. Dang, G. Gour, and B. C. Sanders, arXiv:1706.03828 (2017).

\bibitem{Theurer2017}
T. Theurer, N. Killoran, D. Egloff, and M.B. Plenio, arXiv:1703.10943 (2017).

\bibitem{Brandao2015}
F. G. S. L. Brandão and G. Gour, Phys. Rev. Lett. \textbf{115}, 070503 (2015).

\bibitem{Chitambar2017}
E. Chitambar, J. I. de Vicente, M. Girard, G. Gour, arXiv:1711.03835 (2017).

\bibitem{Coecke2016}
B. Coecke, T. Fritz, R. W. Spekkens, Inf. Comput. \textbf{250}:59-86 (2016).

\bibitem{Gour2017}
G. Gour, Phys Rev. A, \textbf{95}, 062314 (2017).

\bibitem{Nielsen2002}
M. A. Nielsen and I. L. Chuang, \textit{Quantum Computation and Quantum Information} (10th Edition, Cambridge 2011).

\bibitem{Watrous2018}
J. Watrous, \textit{Theory of Quantum Information} (Cambridge 2018).

\bibitem{Bhatia1996}
R. Bhatia, \textit{Matrix Analysis} (Springer, 1996).

\bibitem{Killoran2016}
N. Killoran, F.E.S. Steinhoff and M.B. Plenio, Phys. Rev. Lett. \textbf{116}, 080402 (2016).

\bibitem{Pasieka2009}
A. Pasieka, D. W. Kribs, R. Laflamme, and R. Pereira, Acta Appl. Math. \textbf{108}, 697 (2009).

\bibitem{Gallego2015}
R. Gallego and L. Aolita, Phys. Rev. X \textbf{5}, 041008 (2015).

\bibitem{Vicente2014}
J. I. de Vicente, J. Phys. A \textbf{47}, 424017 (2014).

\bibitem{Amaral2017}
B. Amaral, A. Cabello, M. T. Cunha, and L. Aolita, arXiv:1705.07911 (2017).

\bibitem{Fritz2017}
T. Fritz, Mathematical Structures in Computer Science \textbf{27}, 850 (2017).

\bibitem{Wootters2012}
W. K. Wootters, Found. Phys. \textbf{42}, 19 (2012).

\bibitem{Hardy2012}
L. Hardy and W. K. Wootters, Found. Phys. \textbf{42}, 454 (2012).

\bibitem{Aleksandrova2013}
A. Aleksandrova, V. Borish, and W. K. Wootters, Phys. Rev. A \textbf{87}, 052106 (2013).

\bibitem{Jonathan1999}
D. Jonathan and M. Plenio, Phys. Rev. Lett. \textbf{83}, 3566 (1999).

\bibitem{Aberg2014}
J. Aberg, Phys. Rev. Lett. \textbf{113}, 150402 (2014).

\bibitem{Duarte2016}
C. Duarte, R. C. Drumond, and M. T. Cunha, J. Phys. A \textbf{49}, 145303 (2016).

\end{references}
\end{document}